\documentclass[a4paper,11pt]{article}
\input{Preamble.tex}

\begin{document}

\pagenumbering{Roman}

\hypersetup{pageanchor=false}
\title{Courcelle's Theorem in Truly Linear FPT}

\author{Tuukka Korhonen\thanks{University of Copenhagen, Denmark. \texttt{tuko@di.ku.dk}. Supported by the European Union under Marie Skłodowska-Curie Actions (MSCA), project no. 101206430, and by the VILLUM Foundation, Grant Number 54451, Basic Algorithms Research Copenhagen (BARC).}
\and
Daniel Lokshtanov\thanks{Department of Computer Science, University of California Santa Barbara, Santa Barbara, CA, USA. \texttt{daniello@ucsb.edu}. Supported by NSF Grant CCF-2505099.}
\and
Saket Saurabh\thanks{The Institute of Mathematical Sciences, HBNI, Chennai, India. \texttt{saket@imsc.res.in}.}
}


\maketitle

\thispagestyle{empty}

\begin{abstract}
Recently, Bumpus, Downey, Eagling-Vose, Enright, Fellows, Kutner, Larios-Jones, Martin, Rosamond, and Yates defined \emph{Truly Linear FPT} (TLFPT) to be the class of parameterized problems with algorithms running in time $\OO(n) + f(k)$, where $n$ is the input size and $k$ the parameter [arXiv:2606.02492].
They gave several algorithmic techniques for designing TLFPT algorithms, but left parameterization by treewidth open.

In this paper, we give a general method for designing TLFPT algorithms parameterized by treewidth, solving three open problems posed by Bumpus et al.
In particular, we give a TLFPT algorithm for Courcelle's theorem: We show that given an $n$-vertex $m$-edge graph $G$, an integer $k$, and a $\CMSO_2$-formula $\varphi$, we can in time $\OO(n+m) + f(k, \varphi)$ either conclude that the treewidth of $G$ is more than $k$, or check whether $G$ satisfies $\varphi$.
As a part of our algorithm, we give an approximation algorithm for treewidth that runs in time $\OO(n+m)$ and returns a  tree decomposition whose width is at most $2^{\OO(k)}$ times the optimum. 
Our result also implies a TLFPT algorithm for computing the value of treewidth exactly.
\end{abstract}

\begin{textblock}{20}(-0.5, 7.4)
\includegraphics[width=100px]{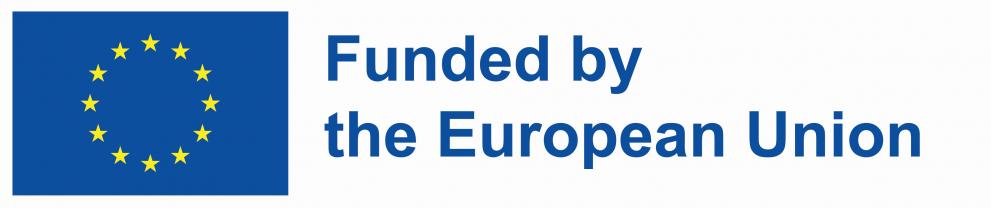}%
\end{textblock}

\thispagestyle{empty}

\newpage


\setcounter{page}{1}



\pagenumbering{arabic}

\hypersetup{pageanchor=true}

\clearpage
\setcounter{page}{1}

\section{Introduction}
\label{sec:intro}
In parameterized algorithms the running time is measured in terms of the input size $n$ and a parameter $k$, which captures some relevant additional information about the input instance.
An algorithm is considered ``good'' if its running time is bounded by $f(k)n^c$, where $f$ is a function independent of $n$ and $c$ is a fixed constant independent of both $n$ and $k$. Such algorithms are called {\em fixed parameter tractable} (FPT).
Not all FPT algorithms are equally good, and much research has been devoted to making $f(k)$ grow as slowly as possible with $k$~\cite{DBLP:journals/talg/CyganNPPRW22,DemaineFHT05jacm,LokshtanovMS11}, or $c$ as small as possible~\cite{bodlaender1993linear, iwata20180, DBLP:conf/focs/KorhonenPS24}.

Beyond optimizing $f(k)$ and $n^c$, one could ask whether the {\em multiplicative} form of the running time is necessary, or whether one could instead obtain algorithms running in time $g(k) + \OO(n^d)$. It is well known, however, that the {\em additive} definition of FPT is equivalent to the multiplicative one: in particular $g(k) + n^d \leq g(k) \cdot n^d$, and similarly $f(k)n^c \leq (f(k))^2 + n^{2c}$.
In a recent foundational paper, Bumpus et al.~\cite{DBLP:journals/corr/abs-2606-02492} observed that this equivalence breaks down once one cares simultaneously about {\em both} the additive-versus-multiplicative distinction and the exponent of $n$. They introduced the class Truly Linear FPT (TLFPT), consisting of all problems solvable in time $\OO(n) + f(k)$ for some function $f$, and proved, via a diagonalization argument, that there exist parameterized problems solvable in time $f(k)n$ but not in TLFPT.
This definition opens up an exciting new research direction: for problems that admit a linear FPT algorithm (one running in time $f(k)n$), can one obtain a TLFPT algorithm? 

Bumpus et al.~\cite{DBLP:journals/corr/abs-2606-02492} gave several techniques for designing TLFPT algorithms, and posed six open problems on the existence of TLFPT algorithms for concrete problems. Three of these would be directly solved by a general method for dynamic programming parameterized by treewidth in TLFPT time.
%
%
In this paper we give such a method,
resolving these three questions and partially resolving a fourth.
Like Bumpus et al.~\cite{DBLP:journals/corr/abs-2606-02492}, we work in the word-RAM model of computation~\cite{savage1998models}.

The standard benchmark for performing dynamic programming on tree decompositions across different computational settings is Courcelle's theorem~\cite{Courcelle90} (see also~\cite{DBLP:journals/jal/ArnborgLS91,DBLP:journals/algorithmica/BoriePT92}), which in a certain sense captures all finite-state dynamic programming algorithms parameterized by treewidth~\cite{DBLP:conf/lics/BojanczykP16}.
Our main result is a TLFPT algorithm for Courcelle's theorem.

\begin{restatable}{theorem}{courcelle}
\label{lem:courcellethm}
There is an algorithm that, given an $n$-vertex $m$-edge graph $G$, an integer $k$, and a $\CMSO_2$-sentence $\varphi$, in time $\OO(n+m) + f(k,\varphi)$, for a computable function $f$, returns either the conclusion that $\tw(G) > k$, or whether $G$ satisfies $\varphi$.
\end{restatable}

The algorithm of \Cref{lem:courcellethm} solves Question~1 of~\cite{DBLP:journals/corr/abs-2606-02492}, which asks ``Is $\MSO_2$ model checking in TLFPT parameterized by BFS-width?''.
As $\MSO_2$ is a restriction of $\CMSO_2$ and BFS-width is a graph parameter whose value is at least treewidth (in fact, at least pathwidth), this solves their question in a higher generality than it was asked.

The main ingredient of the algorithm of \Cref{lem:courcellethm} is the following approximation algorithm for treewidth, which produces a tree decomposition with specific properties that are suitable for the design of TLFPT algorithms in general.

\begin{restatable}{theorem}{twappxthm}
\label{thm:twappx}
There is an algorithm that, given an $n$-vertex $m$-edge graph $G$ and integers $k$ and $s$ with $1 \le s \le n$, in time $\OO(n+m)$ either determines that $\tw(G) > k$, or returns a tree decomposition of $G$, whose leaf bags have size $\le s \cdot 2^{\OO(k)}$, non-leaf bags size $\le 2^{\OO(k)}$, and the number of nodes is $\le n/s$.
\end{restatable}

The idea for designing TLFPT algorithms parameterized by treewidth using \Cref{thm:twappx} is as follows.
Suppose we have a standard dynamic programming algorithm that runs in time $f(w)$ per bag of size $w$.
Now, if we choose $s > f(2^{\OO(k)})$, we can afford to run this dynamic programming algorithm completely normally on the non-leaf bags of the decomposition, as this takes $f(2^{\OO(k)}) \cdot n/s = \OO(n)$ time.
It remains to efficiently find the dynamic programming states for the leaf-bags, which may have size up to  $s \cdot 2^{\OO(k)}$.
Each of them interacts with the rest of the graph only through a boundary of $2^{\OO(k)}$ vertices, so the number of states is still bounded by a function of $k$.
To find the state efficiently, we can pre-compute the state for every possible boundaried graph with $\le s \cdot 2^{\OO(k)}$ vertices in time $g(s,k)$, and fetch it from a global array for each leaf-node.
This results in running time of $\OO(n+m) + g(s,k) = \OO(n+m) + h(k)$.

Implementing the above sketch requires paying attention to low-level details of data structures, but indeed works for Courcelle's theorem.
We believe that it also extends to recovering a solution witness in Courcelle's theorem, and for unweighted optimization variants of Courcelle's theorem~\cite{DBLP:journals/jal/ArnborgLS91}.
However, extending it to weighted problems is not clear, because when weights are taken into account, all leaf bags can induce pairwise non-isomorphic weighted graphs.

Questions~4 and~5 of~\cite{DBLP:journals/corr/abs-2606-02492} asked whether pathwidth and treewidth are TLFPT parameterized by themselves.
Decision versions of these problems follow directly from \Cref{lem:courcellethm} by using the known results that $\tw(G) \le k$ and $\pw(G) \le k$ are $\MSO_2$-expressible graph properties~\cite{DBLP:conf/icalp/LagergrenA91,DBLP:journals/jct/Lagergren98}.
Finding the value in a trivial manner would add a factor of $k$ to the $\OO(n+m)$ in the running time, but the more general statement of our algorithm (\Cref{thm:typecourcelle}) directly works for deciding multiple $\CMSO_2$ properties in a single shot, so we get the following result.

\begin{restatable}{corollary}{twcomp}
\label{cor:twcomp}
There is an algorithm that, given an $n$-vertex $m$-edge graph $G$ and an integer $k$, runs in time $\OO(n+m) + f(k)$, for a computable function $f$, and returns (1) either the value $\tw(G)$ or that $\tw(G) > k$, and (2) either the value $\pw(G)$ or that $\pw(G) > k$.
\end{restatable}

The algorithm of \Cref{cor:twcomp} does not return the corresponding decomposition.
This is not only a fault of the algorithm, but we observe that an explicit representation of an optimal tree decomposition or path decomposition is not always possible in TLFPT space: For the $k \times n$-grid (for $n \ge 8k$), which has treewidth and pathwidth $k$, $kn$ vertices, and $\OO(kn)$ edges, any tree decomposition of width $k$ must have at least $\Omega(kn)$ bags of size at least $k$, i.e., total size at least $\Omega(k^2 n)$ (see \Cref{lem:gridlb}).
However, there are no such obstacles for indirect representations, such as the elimination ordering~\cite{ArnborgCP87} or the elimination forest~\cite{DBLP:journals/lmcs/BojanczykP22} representation, or for explicit representations of approximately optimal tree decompositions.

Question~2 of~\cite{DBLP:journals/corr/abs-2606-02492} asked ``For which families $\{H_i\}_{i \in \mathbb{N}}$ is $H_i$-minor TLFPT parameterized by $|H_i|$?''.
For most such families containing non-planar graphs, it is an open problem to even find a linear FPT algorithm.
Currently, the best running time in general is $f(H) \cdot n^{1+o(1)}$~\cite{DBLP:conf/focs/KorhonenPS24}, while there appears to be no reason to rule out a linear FPT or even a TLFPT algorithm.
For families of planar graphs, \Cref{lem:courcellethm} combined with the grid-minor theorem of Robertson and Seymour~\cite{RobertsonS86a} directly implies a TLFPT algorithm.

\begin{restatable}{corollary}{planarminors}\label{cor:planarminors}
\label{cor:planarminors}
There is an algorithm that, given an $n$-vertex $m$-edge graph $G$ and a set $\mathcal{H} = \{H_1, \ldots, H_\ell\}$ of graphs containing at least one planar graph, in time $\OO(n+m) + f(\mathcal{H})$, for a computable function $f$, returns whether $G$ contains at least one graph from $\mathcal{H}$ as a minor.
\end{restatable}


Using a ``win/win'' scheme similar to that of Corollary~\ref{cor:planarminors}, Theorem~\ref{lem:courcellethm} (or alternatively, by applying Corollary~\ref{cor:planarminors} with carefully chosen forbidden minors ${\cal H}$) directly implies that a number of classic parameterized problems admit truly linear FPT algorithms. 
%
%
%
The approach is as follows. Suppose that we are working with some parameterized problem parameterized by a parameter $k$, such that for every fixed value of $k$ the property under consideration is expressible by a CMSO$_2$ sentence $\varphi_k$. Suppose moreover that there exists a computable function $\eta$ such that, whenever $\operatorname{tw}(G)>\eta(k)$, the answer is already forced to be \textsf{yes} or \textsf{no}.
For such problems we immediately obtain a truly linear FPT algorithm by first running Theorem~\ref{lem:courcellethm} with treewidth bound $\eta(k)$ and formula
$\varphi_k$.  If the algorithm reports that $\operatorname{tw}(G)>\eta(k)$, we
return the forced answer.  
Otherwise we apply the algorithm of Theorem~\ref{lem:courcellethm} to determine whether
$G\models \varphi_k$.  This leads to truly linear FPT algorithms for all problems that satisfy the two assumptions above. 
We remark that this is a standard approach in parameterized algorithms, the only difference with previous algorithms is that we use a truly linear FPT algorithm to compute the treewidth and do the model checking. 
%
%
%
We list several applications of this pattern.

From the textbook~(\cite{DBLP:books/sp/CyganFKLMPPS15}, Chapter 7.7) the approach above applies to {\sc Vertex Cover}, {\sc Feedback Vertex Set}, and {\sc Treewidth-$\eta$-Deletion} for every fixed $\eta$. 
More generally, Demaine and Hajiaghayi~\cite{DBLP:journals/ejc/DemaineH07} showed that this approach applies to every parameter $k$ that (a) is positive for some $g \times g$ grid, (b) is at least the sum over the connected components of a disconnected graph, and (c) admits an FPT algorithm parameterized by the treewidth of the input graph and $k$. The exact same proof gives Truly Linear FPT algorithms for all problems in the framework of Demaine and Hajiaghayi~\cite{DBLP:journals/ejc/DemaineH07}, but with the algorithmic requirement (c) replaced by (c') for every fixed value of $k$ the property under consideration is expressible by a CMSO$_2$ sentence $\varphi_k$.
This immediately yields truly linear FPT algorithms for packing at least $k$ vertex-disjoint minor models of graphs from any fixed finite family ${\cal F}$ of graphs containing at least one planar graph. Most prominently it yields a truly linear FPT algorithm for {\sc Cycle Packing}. 
Finally it is well known that having a cycle of length at least $k$ or a path of length at least $k$ is expressible in MSO$_2$, and that every graph of treewidth at least $k$ contains both a cycle of length $k$ and a path of length $k$~\cite{downey2013fundamentals, fellows1989search}.
This yields truly linear FPT algorithms for~{\sc Long Path} and {\sc Long Cycle}.
We remark that Truly Linear FPT algorithms for {\sc Vertex Cover} and {\sc Long Path} were already obtained by Bumpus et al.~\cite{DBLP:journals/corr/abs-2606-02492}.

\paragraph{Sketch of the proof.}
We already sketched the proof of \Cref{lem:courcellethm}, assuming \Cref{thm:twappx}. So let us sketch here the proof of \Cref{thm:twappx}.

We focus first on the case of $s = 1$, i.e., $2^{\OO(k)}$-approximating treewidth in $\OO(n+m)$ time.
It is known that treewidth can be $2$-approximated in time $2^{\OO(k)} n$~\cite{DBLP:conf/focs/Korhonen21}.
Therefore, our goal is to shrink the size of the input graph $G$ by a factor of $2^{ck}$, for an appropriate constant $c$, in a way that preserves approximation, and then apply the $2$-approximation algorithm.
An idea for this would be to partition the vertex set of $G$ into connected subgraphs of size between $2^{ck}$ and $2^{ck+1}$, and contract each of them.
This does not increase treewidth, reduces treewidth by at most a factor of $2^{ck+1}$, and reduces the number of vertices by at least a factor of $2^{ck}$.
Moreover, a tree decomposition of the contracted graph can be lifted to a tree decomposition of the original graph by uncontracting.

The first issue in this approach is that there are graphs of bounded treewidth for which no such partition into connected subgraphs exists.
For example, stars do not have such a partition.
However, for stars, it is possible to contract disconnected subgraphs in a way that does not increase treewidth: We can merge leaves with each other without increasing treewidth.
We show that this idea generalizes from stars to arbitrary graphs of small treewidth: We can partition the vertex set of any $n$-vertex graph of treewidth $\le k$ into at most $n/2^{\Omega(k)}$ parts of size at most $2^{\OO(k)}$, so that contracting the parts does not increase treewidth.
Moreover, we give an $\OO(n+m)$ time algorithm for finding such a partition or the conclusion $\tw(G) > k$.

The above sketch works for $2^{\OO(k)}$-approximating the value of treewidth.
However, it does not give the corresponding tree decomposition in $\OO(n+m)$ time, because naively uncontracting can increase its total size to more than $\OO(n+m)$.
We solve this with the same technique as with which we introduce the parameter $s$ in \Cref{thm:twappx}.
In particular, for a parameter $s \ge 1$, we reduce the number of nodes by a factor of $s$, while increasing the sizes of only leaf bags by a factor of $s$.
This is done by rather standard tree partitioning techniques, with the idea of partitioning the decomposition tree into $\approx n/s$ connected subtrees of size $\approx s$ by removing $\approx n/s$ nodes, and letting connected subtrees be the large leaves while forming the non-leaf nodes from the removed nodes.
By an appropriate choice of $s$, we can ensure that uncontracting keeps the total size of the tree decomposition $\OO(n+m)$.

\section{Preliminaries}
We discuss preliminaries on graph theory, and present a formulation of Courcelle's theorem in terms of types and boundaried graphs. For an integer $n$, we denote by $[n]$ the set $\{1,\ldots,n\}$, which is the empty set when $n \le 0$. We assume the standard word-RAM model with words of length $\Theta(\log n)$, where $n$ is the input size.

\subsection{Graphs}
The set of vertices of a graph $G$ is denoted by $V(G)$ and the set of edges by $E(G)$.
We denote $|G| = |V(G)|+|E(G)|$.
For a graph $G$ and set $S \subseteq V(G)$ the subgraph of $G$ {\em induced by} $S$ is denoted by $G[S]$ and defined as the graph with vertex set $S$ and edge set $\{uv \in E(G) ~:~ \{u,v\} \subseteq S \}$. {\em Deleting} the vertex set $S$ from a graph $G$ results in the graph $G - S = G[V(G) \setminus S]$.
The {\em neighborhood} of a vertex $u \in V(G)$ is defined as $N(u) = \{v \in V(G) ~:~ uv \in E(G)\}$. The neighborhood of a vertex set $S$ is defined as $N(S) = \left( \bigcup_{u \in S} N(u) \right ) \setminus S$.
When representing graphs, we assume that $V(G)$ is a subset of the positive integers, and the graph is given by a linked list containing $V(G)$ and a linked list containing $E(G)$, with each edge represented as a pair of integers.
We will throughout work with graphs where $V(G) \subseteq [n]$, where $n$ is the original input size, so we assume that each integer representing a vertex fits in a single word.


A {\em tree decomposition} of a graph $G$ is a pair $\cT = (T, \bag)$ where $T$ is a tree and $\bag$ is a function that assigns to each node $t$ of $T$ a set $\bag(t)$ of vertices in $G$ such that the two following conditions are satisfied: {\em (i)} for every vertex $v \in V(G)$ the set $\{t \in V(T) : v \in \bag(t)\}$ is non-empty and induces a connected subtree of $T$, and {\em (ii)} for every edge $uv \in E(G)$ there exists a node $t \in V(T)$ such that $\{u,v\} \subseteq \bag(t)$.
A {\em rooted} tree decomposition is a tree decomposition $\cT = (T, \bag)$ where $T$ is a rooted tree. The {\em descendants} of a node $t$ in a rooted tree, denoted by $\desc(t)$, are the set of nodes of the subtree of $T$ rooted at $t$ (including $t$ itself).
The {\em width} of a tree decomposition $\cT = (T, \bag)$ is defined as $\max_{t \in V(T)} |\bag(t)| - 1$, and the {\em treewidth} of a graph $G$, denoted by $\tw(G)$, is defined as the minimum width of a tree decomposition of $G$.
The {\em size} of a tree decomposition is denoted by $|\cT|$ and defined as $|\cT| = |V(T)| + \sum_{t \in V(T)} |\bag(t)|$. In other words, the size of a tree decomposition is equal to (up to constant factors) the number of machine words needed to describe the tree decomposition in the obvious way where one lists all the nodes and edges of $T$ and the contents of $\bag(t)$ for every $t \in V(T)$.

An {\em elimination ordering} of a graph $G$ is an ordering of its vertices as $v_1, \ldots, v_n$. The {\em filled graph} resulting from $G$ and the elimination ordering is the graph $H$ with vertex set $V(G)$, such that for every pair $1 \leq i < j \leq n$ of integers there is an edge from $v_i$ to $v_j$ in $H$ if and only if there is a path between $v_i$ and $v_j$ in $G[\{v_1, \ldots, v_i\} \cup \{v_j\}]$.
The {\em width} of an elimination ordering of $G$ is the maximum clique size of the resulting filled graph $H$, minus one. It is well known (see e.g.~\cite{arnborg1985efficient}) that the treewidth of $G$ is equal to the minimum width of an elimination ordering of $G$.


For a vertex set $S$, {\em contracting} $S$ produces the graph $G \contr S$ obtained from $G$ by deleting $S$ and adding a new vertex $v_S = \min_{i \in S} i$ adjacent to $N(S)$. Contracting an edge $uv \in E(G)$ is defined as contracting the set $\{u, v\}$. Observe that contracting a set $\{v\}$ containing a single vertex $v$ leaves $G$ unchanged. 
We note that the contraction operation as defined in this paper allows contracting sets $S$ that do not necessarily induce connected subgraphs. 
For sets $S$ that do induce connected subgraphs our definition of contraction and the standard one coincide. 
For a partition $\cP$ of $V(G)$, contracting $\cP$ produces the graph $G \contr \cP$, obtained from $G$ by contracting each of the sets in $\cP$.
%

\subsection{$\CMSO_2$}
We use the standard definitions of Counting Monadic Second-order logic ($\CMSO_2$) on graphs.
We refer to~\cite{CEbook} for an extensive introduction, but recall here the basics and introduce our notation.

Formulas in $\CMSO_2$ have variables of four sorts: vertices, edges, vertex sets, and edge sets.
There are atomic formulas for (1)~testing the equality of two variables of the same sort, (2)~set inclusion, (3)~testing if an edge is incident to a vertex, and (4)~for all integers $a,m$ with $0 \le a < m$, testing whether the cardinality of a set is $a$ modulo $m$.
A $\CMSO_2$-formula is built from these atomic formulas with the connectives $\wedge$, $\vee$, and $\lnot$, and with existential and universal quantifiers.
A $\CMSO_2$-sentence is a $\CMSO_2$-formula without free variables.

The quantifier rank of a $\CMSO_2$-formula is the maximum number of nested quantifiers.
It is known (see~\cite{DBLP:journals/algorithmica/BoriePT92}) that there is a computable function $f(n,r)$, so that any $\CMSO_2$-formula $\varphi$ with at most $n$ free variables and quantifier rank at most $r$ is logically equivalent to a $\CMSO_2$-formula $\varphi'$ of length at most $f(n,r)$ and quantifier rank at most $r$.
Furthermore, $\varphi'$ is computable given $\varphi$.
We denote by $\formulas^{n,r,p}$ the set of $\CMSO_2$-formulas with at most $n$ free variables, each being a vertex variable and having its name from the set $\{x_1, \ldots, x_n\}$, quantifier-bound variables named from the set $\{y_1, \ldots, y_{f(n,r)}\}$, quantifier rank at most $r$, length at most $f(n,r)$, and each modular counting formula having modulus $m \le p$. 
Now, $\formulas^{n,r,p}$ is finite and computable given $n$, $r$, and $p$.
The set $\types^{n,r,p}$ is the powerset of $\formulas^{n,r,p}$.


\subsection{Boundaried graphs}
For an integer $b \ge 0$, a $b$-boundaried graph is a pair $\bg{G} = (G,\bd)$, where $G$ is a graph and $\bd$ is an injective partial function $\bd \colon [b] \pto V(G)$.
We denote by $\dom(\bd)$ the domain of $\bd$, i.e., the subset of $[b]$ for which $\bd$ is defined.
When discussing (non-boundaried) graphs in a context where boundaried graphs are expected, we view them as $0$-boundaried graphs.%

An {\em isomorphism} between two $b$-boundaried graphs $\bg{G}_1 = (G_1,\bd_1)$ and $\bg{G}_2 = (G_2,\bd_2)$ is a bijection $\phi : V(G_1) \rightarrow V(G_2)$ such that (1) $uv \in E(G_1)$ if and only if $\phi(u)\phi(v) \in E(G_2)$, and (2)~for all $u \in V(G_1)$ and $i \in [b]$, $\bd_1(i) = u$ if and only if $\bd_2(i) = \phi(u)$.  Two boundaried graphs are isomorphic if there is an isomorphism between them.

\paragraph{Gluing and permutation.}
For two $b$-boundaried graphs $\bg{G_1} = (G_1, \bd_1)$ and $\bg{G_2} = (G_2, \bd_2)$, the \emph{gluing} of $\bg{G_1}$ and $\bg{G_2}$, denoted by $\bg{G_1} \oplus \bg{G_2}$, is the $b$-boundaried graph $\bg{G_3} = (G_3, \bd_3)$ obtained as follows: We first construct $G_3$ by taking the disjoint union of $G_1$ and $G_2$, and for each $i \in \dom(\bd_1) \cap \dom(\bd_2)$ unifying the vertices $\bd_1(i)$ and $\bd_2(i)$.
Then, $\bd_3$ is constructed by, for each $i \in \dom(\bd_1) \cap \dom(\bd_2)$, setting $\bd_3(i) = \bd_1(i) = \bd_2(i)$, for each $i \in \dom(\bd_1) \setminus \dom(\bd_2)$, $\bd_3(i) = \bd_1(i)$, and for each $i \in \dom(\bd_2) \setminus \dom(\bd_1)$, $\bd_3(i) = \bd_2(i)$.

For a $b$-boundaried graph $\bg{G} = (G, \bd)$ and an injective partial function $f \colon [b] \pto [b]$, the $f$-permutation of $\bg{G}$, denoted by $\per_f(\bg{G})$, is the boundaried graph $\per_f(\bg{G}) = (G, \per_f(\bd))$, where $\per_f(\bd)(i) = \bd(f^{-1}(i))$ for all $i \in f(\dom(\bd))$, and $\per_f(\bd)(i)$ is undefined for other $i$.
Note that $|\dom(\per_f(\bd))| \le |\dom(\bd)|$.

\paragraph{Types.}
For a $b$-boundaried graph $\bg{G} = (G,\bd)$ with $\dom(\bd) = \{i_1, \ldots, i_\ell\}$, and a $\CMSO_2$-formula $\varphi(x_{i_1}, \ldots, x_{i_\ell})$, where $x_{i_1}, \ldots, x_{i_\ell}$ are free vertex-variables, we define that $\bg{G} \models \varphi$ if $G \models \varphi(\bd(i_1), \ldots, \bd(i_\ell))$.
Now, we define that for a $b$-boundaried graph $\bg{G}$, its $(r,p)$-type is $\type^{r,p}(\bg{G}) = \{\varphi \in \formulas^{b,r,p} \colon \bg{G} \models \varphi\}$.
Note that $\type^{r,p}(\bg{G}) \in \types^{b,r,p}$.

\subsection{Courcelle's theorem}
The following two lemmas give a formulation of Courcelle's theorem in terms of boundaried graphs and types. The first lemma states that when gluing two boundaried graphs, the type of the resulting boundaried graph is a function of the types of the two terms. 


\begin{lemma}[See~\cite{gartland2020finding} Proposition 8, and~\cite{grohe2009methods} Lemma 6.1]
\label{lem:courcellejoin}
For all $b,r,p \ge 0$, there exists a function $\oplus \colon \types^{b,r,p} \times \types^{b,r,p} \to \types^{b,r,p}$, computable given $b,r,p$, so that for all $b$-boundaried graphs $\bg{G_1}$ and $\bg{G_2}$ it holds that
\[\type^{r,p}(\bg{G_1}) \oplus \type^{r,p}(\bg{G_2}) = \type^{r,p}(\bg{G_1} \oplus \bg{G_2}).\]
\end{lemma}

The second lemma states that when re-labeling the boundary vertices of a boundaried graph, the type of the resulting boundaried graph is a function only of the type of the initial graph and the applied permutation on boundary labels. 
The lemma immediately follows from the definition of the type of a boundaried graph, together with the observation that, for every $b$-boundaried graph $\bg{G}$ and $\CMSO_2$-formula $\varphi(x_{i_1}, \ldots, x_{i_\ell})$, where $x_{i_1}, \ldots, x_{i_\ell}$ are free vertex-variables, $\bg{G} \models \varphi$ if and only if $\per_f(\bg{G}) \models \varphi'$ where $\varphi'$ is obtained from $\varphi$ by re-labeling the free vertex variables $x_{i_1}, \ldots, x_{i_\ell}$ according to $f$.

\begin{lemma}
\label{lem:courcellepermute}
For all $b,r,p \ge 0$, and each injective partial function $f \colon [b] \pto [b]$, there exists a function $\per_f \colon \types^{b,r,p} \to \types^{b,r,p}$, computable given $b,r,p$, so that for all $b$-boundaried graphs $\bg{G}$ it holds that
\[\per_f(\type^{r,p}(\bg{G})) = \type^{r,p}(\per_f(\bg{G})).\]
\end{lemma}

\section{Computing a tree decomposition}
In this section we prove \Cref{thm:twappx}, which we now re-state.

\twappxthm*

In order to prove \Cref{thm:twappx}, we  start by proving several lemmas.
The first lemma is the fact that graphs of bounded treewidth have linear neighborhood complexity.
This is well-known (e.g. \cite{DBLP:journals/combinatorica/JoretR24}), but we present a self-contained proof that also obtains the optimal exponential dependence on treewidth.

\begin{lemma}
\label{lem:nbcompl}
Let $G$ be a graph of treewidth $k$, and $A \subseteq V(G)$ a non-empty set.
The number of distinct sets $N(v) \cap A$ with $v \in V(G)$ is at most $2^{k+2} |A|$.
\end{lemma}
\begin{proof}
Let
\[
 {\cal F}=\{N_G(v)\cap A : v\in V(G)\setminus A\}.
\]
For every $S\in{\cal F}$ choose one representative $x_S\in V(G)\setminus A$ such that
$N_G(x_S)\cap A=S$, and let $X$ be the set of representatives.  Let $B$ be the bipartite
graph with bipartition $(X,A)$ and edge set
\[
 \{xa : x\in X,\ a\in A,\ a\in N_G(x)\}.
\]
Then $B$ is a subgraph of $G[X\cup A]$, and hence $\operatorname{tw}(B)\le k$.

Fix an elimination ordering of $B$ of width at most $k$, and let $B'$ be the filled graph obtained from this elimination ordering.  Thus every vertex has at most $k$ forward neighbors in $B'$, and the forward neighborhood of every vertex is a clique in $B'$.

Call a vertex $x\in X$ {\em bad} if it is a forward neighbor in $B'$ of some vertex of $A$.
Since every $a\in A$ has at most $k$ forward neighbors, the number of bad vertices is at
most $k|A|$.

Now consider a vertex $x\in X$ which is not bad.  If $a\in N_B(x)\cap A$, then $a$ cannot
appear before $x$ in the elimination ordering, because then $x$ would be a forward
neighbor of $a$, making $x$ bad.  Hence every vertex of $N_B(x)\cap A$ is a forward
neighbor of $x$ in $B'$.  Therefore $N_B(x)\cap A$ is a subset of the forward neighborhood
of $x$ in $B'$, and hence it is a clique in $B'[A]$.

It remains to bound the number of cliques in $B'[A]$.  Every nonempty clique $C$ of
$B'[A]$ is charged to its first vertex $a$ in the elimination ordering.  Then
$C\setminus\{a\}$ is contained in the forward neighborhood of $a$, which has size at most
$k$.  Thus for each $a\in A$ there are at most $2^k$ choices for such a clique, and so
$B'[A]$ has at most $2^k|A|$ nonempty cliques.

We can now complete the proof. We have that
\[
 |{\cal F}| \le k|A|+2^k|A|+1.
\]
Finally, vertices of $A$ themselves contribute at most $|A|$ distinct sets $N(v) \cap A$.  Therefore the number of distinct sets $N_G(v)\cap A$ with $v\in V(G)$ is at most
\[
 k|A|+2^k|A|+1+|A| \le 2^{k+2}|A|,
\]
as claimed.
\end{proof}

Then we show that connected graphs with small degree can be partitioned into connected parts of approximately equal size.
This lemma can be considered folklore, but we give a self-contained proof.

\begin{lemma}
\label{lem:treechop}
There is an algorithm that, given a connected graph $G$ with maximum degree $\Delta$ and a parameter $h$ with $|V(G)| \ge h$, in time $\OO(n+m)$ finds a partition $\cP$ of $V(G)$ so that for each $P \in \cP$:
\begin{itemize}
\item $G[P]$ is connected, and
\item $h \le |P| \le h \Delta$.
\end{itemize}
\end{lemma}
\begin{proof}
We can assume without loss of generality that $G$ is a tree.
The cases of $\Delta \le 1$ and $h \le 1$ are trivial, so we can assume $\Delta \ge 2$ and $h \ge 2$.

To prove that such a partition exists, we proceed via induction.
The base case is that $G$ has at most $h \Delta$ vertices, in which case we can take a single set $P = V(G)$.
Then, if $G$ has more than $h \Delta$ vertices, suppose that $G$ is a rooted tree, and let $v$ be a vertex, so that the subtree rooted at $v$ has $\ge h$ vertices, but all subtrees rooted under $v$ have $< h$ vertices.
Because $v$ has degree $\le \Delta$, the subtree rooted at $v$ has $\le (h-1) \Delta + 1 < h \Delta$ vertices, implying that it is not the root and in fact has at most $\Delta-1$ children, so it has $\le (h-1) (\Delta-1)+1 \le h (\Delta-1)$ vertices.
After removing this subtree as one set of $P$, the remaining tree is connected and has $\ge h$ vertices, so the induction assumption can be applied to it.

The above induction proof can be implemented in $\OO(n)$ time via depth-first-search.
\end{proof}

The next lemma encapsulates the key step of our algorithm that reduces the size of the input graph.
In particular, the first step of our algorithm is to contract sets of vertices of size $2^{\OO(k)}$ so that the size of the graph is reduced by a factor of $2^{\Omega(k)}$, but treewidth is not increased.
This is done by the following algorithm.

\begin{lemma}
\label{lem:findpartition}
There is an algorithm that, given an $n$-vertex $m$-edge graph $G$ and integers $k$ and $h$, in time $\OO(n+m)$ either concludes that $\tw(G) > k$, or finds a partition $\cP$ of $V(G)$ so that
\begin{itemize}
\item $|V(G \contr \cP)| \le |V(G)|/h$,
\item $\tw(G \contr \cP) \le \tw(G)$, and
\item for all $P \in \cP$, $|P| \le h^3 \cdot 2^{\OO(k)}$.
\end{itemize}
\end{lemma}
\begin{proof}
First, if $m > k n$, we can conclude that $\tw(G) > k$.
Assume then that $m \le k n$.

Let $d = h^2 \cdot k 2^{k+6}$, and let $H \subseteq V(G)$ be the vertices with degree $\ge d$.
We have that $|H| \le \frac{2m}{d} \le \frac{n}{h^2 2^{k+5}}$.
The vertices in $H$ will be placed as singleton parts to the partition $\cP$.
We compute the connected components of $G \setminus H$, and say that a component $C$ is \emph{large} if $|C| \ge 2h$, and \emph{small} otherwise.

For each large component $C$, we let $\cP_C$ be the partition of $C$ obtained by applying \Cref{lem:treechop} with the parameters $2h$ and $d$.
In particular, each part in $\cP_C$ has at least $2h$ and at most $2hd \le h^3 \cdot k 2^{k+7}$ vertices.
We note that contracting all such parts does not increase the treewidth of $G$, because they are connected subgraphs.

We group the small components $C$ by their neighborhoods $N(C)$.
As the sum of $|N(C)|$ is at most $m$, this can be done by radix sort in $\OO(n+m)$ time.
By applying \Cref{lem:nbcompl} after contracting each component $C$, we get that the number of neighborhoods is $\le 2^{\tw(G)+2} |H|$.
Therefore, if it is more than $2^{k+2} |H|$, we return that $\tw(G) > k$.

Let $X$ be a neighborhood and $\mathcal{C}_X$ the collection of small components $C$ with $N(C) = X$.
We say that $X$ is a \emph{small neighborhood} if $\sum_{C \in \mathcal{C}_X} |C| \le 2h$ and \emph{large neighborhood} otherwise.
The total number of vertices in components with small neighborhoods is at most $2h \cdot 2^{k+2}|H| \le 2h \cdot 2^{k+2} \frac{n}{h^2 2^{k+5}} \le \frac{n}{4h}$.
We will ignore those vertices, i.e., place them as singletons to $\cP$.

Let $X$ be a large neighborhood.
Now, for each $C \in \mathcal{C}_X$ we have $|C| \le 2h$, but $\sum_{C \in \mathcal{C}_X} |C| > 2h$.
Therefore, we can group components in $\mathcal{C}_X$ so that each group has between $2h$ and $6h$ vertices.
We put these groups to the partition $\cP$.
Even though these groups do not necessarily form connected subgraphs, their contraction does not increase treewidth because their contraction corresponds to contracting each component in the group, and then deleting all but one of the resulting vertices.

It remains to show that $|V(G \contr \cP)| \le |V(G)|/h$.
Each part in $\cP$ is either a singleton vertex, or has at least $2h$ vertices.
The number of singleton parts is at most $|H|$ plus the number of vertices in small components with small neighborhoods, which is $|H| + \frac{n}{4h} \le \frac{n}{2h}$.
The number of parts with at least $2h$ vertices is at most $\frac{n}{2h}$.
Therefore, the total number of parts is at most $n/h$.
\end{proof}

Next we give another folklore lemma about simplifying tree decompositions.

\begin{lemma}
\label{lem:simplifytd}
There is an algorithm that, given a tree decomposition $\cT$ of an $n$-vertex graph $G$ of width $k$, in time $\OO(|\cT| + nk)$ returns a tree decomposition $\cT' = (T',\bag')$ of $G$ that has width $k$, $|V(T')| \le 2n$, and maximum degree $3$.
\end{lemma}
\begin{proof}
We first describe transforming $\cT$ into having at most $n$ nodes, and then transforming that into maximum degree $3$.

Denote $\cT = (T, \bag)$.
We root $T$ at an arbitrary node $r \in V(T)$ so that $\bag(r)$ is non-empty.
Consider then the following process: As long as $\cT$ contains a node $t$ with parent $p$ so that $\bag(t) \subseteq \bag(p)$, we contract the edge $tp$, and let the resulting node have bag equal to $\bag(p)$.
Let $\cT' = (T',\bag')$ be the resulting tree decomposition.
Obviously, $\cT'$ is a tree decomposition of $G$ and has width at most $k$.
It remains to prove that the contraction process can be implemented in $\OO(|\cT|)$ time, and that the resulting tree decomposition has at most $n$ nodes.

We say that the \emph{home-node} of a vertex $v$ is the node $t$ of $T$ closest to the root such that $v \in \bag(t)$.
We note that $\cT'$ is obtained from $\cT$ by contracting each node that is not a home-node of any vertex to its parent.
Because at most $n$ nodes can be home-nodes, it follows that $|V(T')| \le n$.
Furthermore, $\cT'$ is easy to compute in $\OO(|\cT|)$ time after finding which nodes are home-nodes of at least one vertex.
Such nodes can be found in $\OO(|\cT|)$ time by recording for each vertex $v$ the depth of the lowest-depth node whose bag contains $v$.

Finally, we transform the resulting tree decomposition into maximum degree $3$ by replacing high-degree nodes by binary trees.
This runs in time $\OO(nk)$ and increases the number of nodes at most by a factor of $2$.
\end{proof}

The following lemma is the other key ingredient of \Cref{thm:twappx}.
It shrinks the number of bags of a tree decomposition by a given factor $h$, while increasing the sizes of the leaf bags by a factor of $\OO(h)$ and the sizes of the non-leaf bags by only a constant factor.

\begin{lemma}
\label{lem:compresstd}
There is an algorithm that, given a tree decomposition $\cT$ of an $n$-vertex graph $G$ of width $k$ and an integer $h$, in time $\OO(|\cT| + nk)$ returns a tree decomposition $\cT' = (T', \bag')$ of $G$ so that
\begin{itemize}
\item each leaf bag of $\cT'$ has size at most $\OO(h k)$,
\item each non-leaf bag of $\cT'$ has size at most $2k+2$, and
\item $|V(T')| \le n/h$.
\end{itemize}
\end{lemma}
\begin{proof}
We start by applying \Cref{lem:simplifytd}, after which we assume that the input tree decomposition $\cT = (T,\bag)$ has $|V(T)| \le 2n$ and maximum degree $3$.

\begin{claim}
\label{lem:compresstd:claim1}
For an integer $p \ge 1$, we can find in $\OO(n)$ time a set $X \subseteq V(T)$ of nodes of size $|X| \le 4n/p$, so that $T \setminus X$ has $\le 12n/p$ connected components, and each connected component $C$ of $T \setminus X$ has $|C| \le 3p$ nodes and at most $2$ neighbors in $X$.
\end{claim}
\begin{claimproof}
We apply \Cref{lem:treechop} to $T$ with the parameters $3$ and $p$, and obtain a partition $\cP$ of $V(T)$ into connected parts so that each part has at least $p$ and at most $3p$ nodes.
We root $T$ at an arbitrary node, and construct a set $Y$ by taking the lowest-depth node of each part $P \in \cP$.
We have that $|Y| \le 2n/p$ and each connected component of $T \setminus Y$ is a subset of a part $P \in \cP$.

Then, we construct $X$ by taking the \emph{LCA-closure} of $Y$, that is, $X$ is the superset of $Y$ constructed by adding each node of $T$ that is the least common ancestor of a pair of nodes from $Y$.
We have that $|X| \le 2|Y| \le 4n/p$.
Furthermore, this guarantees that each component of $T \setminus X$ has at most $2$ neighbors in $X$.

Because the maximum degree of $T$ is $3$, the number of connected components of $T \setminus X$ is at most $3|X| \le 12n/p$.
\end{claimproof}

We apply \Cref{lem:compresstd:claim1} with $p = 28h$.
To construct $\cT' = (T', \bag')$, we first contract each connected component $C$ of $T \setminus X$ into a single node $v_C$.
Then, for each such node we add an adjacent leaf $v_C'$.
For each node $t \in X$, we set $\bag'(t) = \bag(t)$.
For each node $v_C$, we set $\bag'(v_C)$ to be the union of the bags of the adjacent nodes in $X$.
As there are at most two of such nodes, $|\bag'(v_C)| \le 2 (k+1)$.
For each node $v_C'$, we set $\bag'(v_C')$ to be the union of the bags of the nodes in $C$, i.e., $\bag'(v_C') = \bigcup_{u \in C} \bag(u)$.
As $|C| \le 3p$ and $p = 28h$, we have $|\bag'(v_C')| \le \OO(h k)$.

This construction is indeed a tree decomposition of $G$, and satisfies the two required bag size bounds.
It remains to prove that $|V(T')| \le n/h$.
We have that $|V(T')|$ is $|X|$ plus two times the number of connected components of $T \setminus X$.
Therefore, $|V(T')| \le 4n/p + 24 n/p \le n/h$.
\end{proof}

We also need the following algorithm for computing treewidth.

\begin{theorem}[\cite{DBLP:conf/focs/Korhonen21}]
\label{thm:tw2appx}
There is an algorithm that, given an $n$-vertex $m$-edge graph $G$ and an integer $k$, in time $2^{\OO(k)} n + \OO(m)$ either returns a tree decomposition of $G$ of width $\le 2k+1$, or determines that $\tw(G) > k$.
\end{theorem}

Now we are ready to put all of the above ingredients together to prove \Cref{thm:twappx}.

\twappxthm*
\begin{proof}
Let $f(k)$ be a function in $2^{\OO(k)}$ so that the algorithm of \Cref{thm:tw2appx} runs in time $f(k) n + \OO(m)$.
We assume that $f(k) > \max(k,4)$.
We start by applying the algorithm of \Cref{lem:findpartition} with the parameters $k$ and $h = f(k)$.
It either concludes that $\tw(G) > k$, or finds a partition $\cP$ of $V(G)$ so that $|\cP| \le n/h$, $\tw(G \contr \cP) \le \tw(G)$, and for all $P \in \cP$, $|P| \le h^3 \cdot 2^{\OO(k)} \le 2^{\OO(k)}$.
Let us denote by $g(k)$ the bound such that $|P| \le g(k)$.

Let $G' = G \contr \cP$.
We can construct $G'$ in $\OO(n+m)$ time with the help of radix sort.
Note that $|E(G')| \le m$.
We use the algorithm of \Cref{thm:tw2appx} to either conclude that $\tw(G') > k$, in which case we can conclude $\tw(G) > k$, or find a tree decomposition $\cT$ of $G'$ of width $\le 2k+1$.
It runs in time $f(k) |V(G')|+\OO(|E(G')|) \le \OO(n+m)$.
Note that this implies that $|\cT| \le \OO(n+m)$.

We apply the algorithm of \Cref{lem:compresstd} with the tree decomposition $\cT$ and the parameter $h = \max(s, k \cdot g(k))$.
It runs in time $\OO(|\cT| + |V(G')|k) = \OO(n+m)$, and returns a tree decomposition $\cT' = (T',\bag')$ of $G'$, whose leaf bags have size $\le \OO(h k)$, non-leaf bags have size $\le \OO(k)$, and which has $|V(T')| \le |V(G')|/h \le n/s$.

Now, we construct a tree decomposition $\cT''$ of $G$ from $\cT'$ by replacing each vertex corresponding to a part $P \in \cP$ by the set of vertices $P$.
It is easy to see that this indeed yields a tree decomposition of $G$, whose leaf bags have size $\le \OO(h \cdot k \cdot g(k)) \le s \cdot 2^{\OO(k)}$ and non-leaf bags size $\le \OO(k \cdot g(k)) \le 2^{\OO(k)}$.

The running time of this replacement operation is linear in the size $|\cT''|$ of the resulting tree decomposition, so it suffices to bound $|\cT''|$.
Each vertex that occurs in more than one bag occurs in a non-leaf bag, so we have that $|\cT''| \le |V(T')| + n + |V(T')| \cdot \OO(k \cdot g(k)) \le n + (n/h) \cdot \OO(k \cdot g(k)) \le \OO(n)$.
\end{proof}

\section{$\CMSO_2$ model checking}
The following is a version of \Cref{lem:courcellethm} that is more suitable for different applications.
Note that here we view $G$ as a $0$-boundaried graph.

\begin{restatable}{theorem}{typecourcelle}
\label{thm:typecourcelle}
There is an algorithm that, given an $n$-vertex $m$-edge graph $G$ and integers $k,r,p$, returns either $\tw(G) > k$, or $\type^{r,p}(G)$, in time $\OO(n+m) + f(k,r,p)$, for a computable function $f$.
\end{restatable}

In this section, we first prove \Cref{thm:typecourcelle}, and then use it to derive the consequences mentioned in \Cref{sec:intro}.

\subsection{Proof of \Cref{thm:typecourcelle}}
This subsection is dedicated to the proof of \Cref{thm:typecourcelle}.
We start by proving a central lemma, which encapsulates the use of tabulation in the algorithm of \Cref{thm:typecourcelle}.

\begin{lemma}
\label{lem:tabulation}
Let $n$ be an integer so that the word-length is $\Theta(\log n)$ bits.
There is an algorithm that, given integers $b,r,p,s$, and a list $\bg{G_1}, \ldots, \bg{G_\ell}$ of $b$-boundaried graphs, so that $|V(G_i)| \le s$, $V(G_i) \subseteq [n]$, and $\sum_{i=1}^\ell |G_i| \le 2n$, in time $\OO(n) + \ell \cdot f(b,r,p) + g(b,r,p,s)$ returns $\type^{r,p}(\bg{G_i})$ for all $i \in [\ell]$, where $f$ and $g$ are computable functions.
\end{lemma}
\begin{proof}
First, if $n \le s^{s+2s^2+b}$, we solve the problem by brute-force in time computable in $n, b, r, p, s$, which in this case is bounded by a computable function in $b,r,p,s$.
Therefore, for the remainder of the proof we assume $s^{s+2s^2+b} < n$.

We use counting sort to in time $\OO(n + \sum_{i=1}^\ell |G_i|)$ transform each $\bg{G_i}$ to an isomorphic $b$-boundaried graph whose vertex set is a subset of $[s]$.
This does not change $\type^{r,p}(\bg{G_i})$.

Each such boundaried graph can be represented by a tuple of length $|V(G)|+2 |E(G)| + b$, whose all members are integers in $[s]$, by first listing the set of vertices, then the set of edges, and then the boundary.
We compute such representation for each $\bg{G_i}$ in total $\OO(\sum_{i=1}^\ell |G_i|) = \OO(n)$ time.
We do not care about the representation being canonical in any way, but observe that if two representations are the same, then the corresponding boundaried graphs are isomorphic.
There are at most $s^{s+2s^2+b}$ such representations.

Because $s^{s+2s^2+b} < n$, the representation of $\bg{G_i}$ can be represented by an integer in $[n]$, which can be computed in $\OO(|G_i|)$ time.
Therefore, we can group the boundaried graphs with the same representation in $\OO(n)$ time.

Now it remains to compute the type for only one boundaried graph in each of the groups.
We do this by brute-force, running in time computable in $b,r,p,s$.
As the number of groups is at most $s^{s+2s^2+b}$, the total running time is also computable in $b,r,p,s$.
We note that $\type^{r,p}(\bg{G_i})$ can be represented in word-size that is computable in $b,r,p$, so the output-size is $\ell \cdot f(b,r,p)$, where $f$ is a computable function.
\end{proof}

Let $\cT = (T, \bag)$ be a rooted tree decomposition of a graph $G$.
For a vertex $v \in V(G)$, we define the \emph{home node} of $v$ to be the lowest-depth node $t_v$ with $v \in \bag(t_v)$.
Due to the subtree-property of tree decompositions, $t_v$ is uniquely defined.
Similarly, for an edge $uv \in E(G)$, we define the \emph{home node} of $uv$ to be the lowest-depth node $t_{uv}$ with $\{u,v\} \subseteq \bag(t_{uv})$.
It is not hard to observe that $t_{uv}$ must be the higher-depth node among the two nodes $t_u$ and $t_v$, so it is indeed uniquely defined.

We say that the \emph{edge-annotation} of $\cT$ is the function $\edges \colon V(T) \to 2^{E(G)}$, that maps each node $t$ to the set of edges $uv$ for which $t_{uv} = t$.

\begin{lemma}
\label{lem:edgeannotate}
Given a graph $G$ and a rooted tree decomposition $\cT = (T,\bag)$ of $G$, the edge-annotation of $\cT$ can be computed in $\OO(|G|+|\cT|)$ time.
\end{lemma}
\begin{proof}
First, we compute $t_v$ for all $v \in V(G)$ in time $\OO(|\cT|)$ by depth-first search.
At the same time, we can compute the depth of each node.
Therefore, for each edge $uv$, we can find $t_{uv}$ in constant time by taking the higher-depth of the two nodes $t_v$ and $t_u$.
\end{proof}

Now we are ready to prove \Cref{thm:typecourcelle}.

\typecourcelle*
\begin{proof}
Let $s$ be a positive integer that will be selected during the course of the proof to be large enough, but to depend only on $k$, $r$, and $p$ in a computable manner.
We start by applying the algorithm of \Cref{thm:twappx} with the parameters $G$, $k$, and $s$.
It runs in time $\OO(n+m)$ and either determines that $\tw(G) > k$, or returns a tree decomposition $\cT = (T,\bag)$ of $G$ whose leaf bags have size $\le s \cdot 2^{\OO(k)}$, non-leaf bags size $\le 2^{\OO(k)}$, and the number of nodes is $|V(T)| \le n/s$.

Let $b = 2^{\OO(k)}$ be the maximum size of a non-leaf bag of $\cT$.
By choosing $s \ge b+1$, we have $|\cT| \le n + |V(T)| \cdot (b+1) \le 2n$.

We root $\cT$ at an arbitrary non-leaf node and use \Cref{lem:edgeannotate} to compute the edge-annotation $\edges$ of $\cT$ in $\OO(n+m+|\cT|) = \OO(n+m)$ time.
For each non-leaf node $t$, we define $\bg{G_t}$ to be the $b$-boundaried graph $\bg{G_t} = (G_t, \bd_t)$, with $V(G_t) = \bigcup_{d \in \desc(t)} \bag(d)$, $E(G_t) = \bigcup_{d \in \desc(t)} \edges(d)$, and $\bd_t$ mapping each integer $i \in [|\bag(t)|]$ to the $i$th vertex of $\bag(t)$ in the sorted order (recall that the vertices of $G$ are integers).

For each leaf node $t$ with parent $p$, we define $\bg{G_t}$ to be the $b$-boundaried graph $\bg{G_t} = (G_t, \bd_t)$ with $V(G_t) = \bag(t)$, $E(G_t) = \edges(t)$, and $\bd_t$ mapping each integer $i \in [|\bag(t) \cap \bag(p)|]$ to the $i$th vertex of $\bag(t) \cap \bag(p)$.
Having the $\edges$ function, we can explicitly compute the collection of these boundaried graphs in $\OO(n+m)$ time.

Let the boundaried graphs associated with the leaves be $\bg{G_1}, \ldots, \bg{G_\ell}$.
We note that $\sum_{i=1}^{\ell} |\bg{G_i}| \le |\cT| + m \le 2n+m$, and that $\ell \le n/s$.
We apply the algorithm of \Cref{lem:tabulation} to compute, in time $\OO(n+m)+\ell \cdot f(b,r,p) + g(b,r,p,s)$, for computable functions $f$ and $g$, the type $\type^{r,p}(\bg{G_i})$ for each $i \in [\ell]$.
By choosing $s \ge f(b,r,p)$, the running time is bounded by $\OO(n+m) + g(b,r,p,s)$.

We then compute $\type^{r,p}(\bg{G_t})$ for each non-leaf node $t$ by dynamic programming.
First, in time $|\bag(t)|^{\OO(1)}$, we can sort the vertices of $\bag(t)$ and compute the $b$-boundaried graph $\bg{G_t'}$ with the vertex set $\bag(t)$, edge set $\edges(t)$, and boundary $\bag(t)$ assigned in the sorted order.
Suppose that the children of $t$ are $c_1, \ldots, c_\ell$.
The boundary of $\bg{G_{c_i}}$ is a superset of $\bag(c_i) \cap \bag(t)$.
Let $\bg{G_{c_i}'}$ be the boundaried graph obtained from $\bg{G_{c_i}}$ by restricting the boundary to $\bag(c_i) \cap \bag(t)$, and permuting the indices so that they match to the indices of the same vertices in $\bg{G_t'}$.
We can compute $\type^{r,p}(\bg{G_{c_i}'})$ in time $h_1(b,r,p)$ from $\type^{r,p}(\bg{G_{c_i}})$ with the use of \Cref{lem:courcellepermute}, where $h_1$ is a computable function.
Now, \[\bg{G_t} = \bg{G_t'} \oplus \bg{G_{c_1}'} \oplus \ldots \oplus \bg{G_{c_\ell}'},\]
so we can compute $\type^{r,p}(\bg{G_t})$ in time $\ell \cdot h_2(b,r,p)$ with the use of \Cref{lem:courcellejoin}, where $h_2$ is a computable function.
Therefore, we computed $\type^{r,p}(\bg{G_t})$ in total time $(|\bag(t)|^{\OO(1)} + h_1(b,r,p) + h_2(b,r,p)) \cdot \ell$, where $\ell$ is the number of children, i.e., in time $h(b,r,p) \cdot \ell$, where $h$ is a computable function.

By choosing $s \ge h(b,r,p)$, the total running time spent on the internal nodes is $h(b,r,p) \cdot |V(T)| \le h(b,r,p) \cdot n/s \le \OO(n)$. 
In the end, we obtain $\type^{r,p}(\bg{G_t})$, from which $\type^{r,p}(G)$ is easy to obtain.
The total running time is $\OO(n+m) + f(s, k, r, p)$, for a computable function $f$, where $s$ is bounded by a computable function on $k,r,p$.
\end{proof}

\subsection{Corollaries}
Let us now prove all of the applications of \Cref{thm:typecourcelle} we claimed in \Cref{sec:intro}.
First, we obtain the simple formulation of Courcelle's theorem.

\courcelle*
\begin{proof}
Let $r$ be the quantifier rank of $\varphi$ and $p$ the maximum modulus in modular counting predicates of $\varphi$.
In time computable in $|\varphi|$, we can transform $\varphi$ into a logically equivalent sentence $\varphi'$ in $\formulas^{0,r,p}$.
Then, we use \Cref{thm:typecourcelle} to either conclude $\tw(G) > k$, or compute $\type^{r,p}(G)$.
We have that $G$ satisfies $\varphi$ if and only if $\varphi' \in \type^{r,p}(G)$.
\end{proof}

Then we obtain an algorithm for computing treewidth and pathwidth.

\twcomp*
\begin{proof}
By the result of Lagergren and Arnborg~\cite{DBLP:conf/icalp/LagergrenA91}, for each $w$ there exists a $\CMSO_2$-formula $\varphi^t_w$, computable given $w$, so that $G$ satisfies $\varphi^t_w$ if and only if $\tw(G) \le w$.
By the result of Lagergren~\cite{DBLP:journals/jct/Lagergren98}, for each $w$ there exists a $\CMSO_2$-formula $\varphi^p_w$, computable given $w$, so that $G$ satisfies $\varphi^p_w$ if and only if $\pw(G) \le w$.
We start by computing such formulas $\varphi^t_w$ and $\varphi^p_w$ for all $0 \le w \le k$.
Let $r$ be the maximum quantifier rank in such formulas and $p$ the maximum modulus in the modular counting predicates in such formulas.
We first transfer the formulas to equivalent ones in $\formulas^{0,r,p}$.
Then we run the algorithm of \Cref{thm:typecourcelle} to either compute $\type^{r,p}(G)$ or the conclusion that $\tw(G) > k$, in which case also $\pw(G) > k$.
From $\type^{r,p}(G)$ we obtain for each of the formulas whether $G$ satisfies them, based on which we can return the conclusion.
Everything runs in time $\OO(n+m)+f(k)$, where $f$ is a computable function.
\end{proof}

Then we obtain an algorithm for planar minor testing.

\planarminors*
\begin{proof}
By the grid minor theorem~\cite{RobertsonS86a}, there is a function $g(k)$, so that if a graph $G$ does not contain a planar graph $H$ of size $|H| \le k$ as a minor, then $\tw(G) \le g(k)$.
The function $g$ has computable upper bounds, e.g. polynomial bounds proven in~\cite{DBLP:journals/jacm/ChekuriC16}.

Therefore, if $w = \min_{H_i \text{ is planar}} \{g(|H_i|)\}$, then if $\tw(G) > w$, we can return that $G$ contains at least one graph from $\mathcal{H}$ as a minor.
Given $\mathcal{H}$, we can also compute a sentence $\varphi$ of $\CMSO_2$, which $G$ satisfies if and only if $G$ contains at least one graph from $\mathcal{H}$ as a minor.

We apply the algorithm of \Cref{lem:courcellethm} with $w$ and $\varphi$, and if either $\tw(G) > w$ or $G$ satisfies $\varphi$, return yes, and otherwise return no.
\end{proof}

\section{Lower bound for the size of a tree decomposition}
We show that for all $k$, there are graphs with $n$ vertices, $\OO(n)$ edges, and treewidth $k$, whose all optimum-width tree decompositions $\cT$ have size $|\cT| \ge \Omega(kn)$.
This is tight, since \Cref{lem:simplifytd} implies that every $n$-vertex graph with treewidth $k$ has a tree decomposition of width $k$ and size $\OO(kn)$.

\begin{lemma}
\label{lem:gridlb}
For all positive integers $n$ and $k$ with $n \ge 8k$, every tree decomposition of the $k \times n$ grid graph of width $k$ has at least $\Omega(kn)$ bags of size $\ge k$.
\end{lemma}
\begin{proof}
Let $G$ be the $k \times n$ grid graph.
We assume that $k \ge 2$.
We start by proving properties of separators of size $\le k+1$ of $G$.

For a set $X \subseteq V(G)$ and a connected component $C$ of $G-X$, we say that $C$ is \emph{large} if $C$ intersects all of the $k$ distinct rows of $G$. Otherwise, we say that $C$ is \emph{small}.

\begin{claim}
If $|X| \le k+1$, then at most $k^2$ vertices of $G-X$ are in small components.
\end{claim}
\begin{claimproof}
Let $S \subseteq V(G) \setminus X$ be the union of the small components.
If $S$ would intersect $k+2$ distinct columns of $G$, then because $|X| \le k+1$, one of these columns would not intersect with $X$, implying that $S$ would contain the column entirely, causing $S$ to contain a large component.
It follows that $S$ intersects at most $k+1$ distinct columns.
In fact, $S$ intersects at most $k$ distinct columns, because if $S$ intersected $k+1$ distinct columns, $X$ would occur only in the columns where $S$ occurs, causing $S$ to intersect at least $(n-k-1)/2 > k+1$ distinct columns.

Because $S$ intersects at most $k$ distinct columns, it can have at most $k^2$ vertices.
\end{claimproof}

\begin{claim}
\label{lem:gridlb:claim2}
If $|X| < k$, then $G-X$ has only one large component.
\end{claim}
\begin{claimproof}
There is a row that is disjoint from $X$, so that row is contained in one connected component $C$ of $G-X$, and thus no other connected component can be large.
\end{claimproof}

\begin{claim}
\label{lem:gridlb:claim3}
If $|X| \le k+1$ and $G-X$ has more than one large component, then $G-X$ has exactly two large components and at most one small component.
Furthermore, if there is a small component, then it has only one vertex.
\end{claim}
\begin{claimproof}
First, note that three large components would force $X$ to contain at least two vertices from each row, implying $|X| \ge 2k > k+1$.

Because $G-X$ has two large components, $X$ must intersect all rows of $G$.
It follows that there is at most one row which intersects $X$ in two vertices, and all other rows intersect $X$ in exactly one vertex.
When traveling along a row that intersects $X$ in exactly one vertex, it must first intersect one of the large components, then $X$, and then the other large component.
Therefore, it does not intersect any small component, and thus all small components must be contained within the single row that intersects $X$ in two vertices.
It follows that there can be only one small component.
Furthermore, if that component would contain more than one vertex, then the adjacent row would have to intersect $X$ in at least two vertices.
It follows that the small component has only one vertex.
\end{claimproof}

Let $\cT = (T,\bag)$ be a tree decomposition of $G$ of width $k$.
We edit $\cT$ by subdividing each edge $xy$ of $T$ and adding a new bag $\bag(t_{xy}) = \bag(x) \cap \bag(y)$ on the subdivision node $t_{xy}$.
This increases the number of bags of size $\ge k$ by at most a factor of $2$, since every new such bag can be charged from its child of size $\ge k$.
Therefore, we assume without loss of generality that $\cT$ satisfies the property that for all $xy \in E(T)$, either $\bag(x) \subseteq \bag(y)$ or $\bag(y) \subseteq \bag(x)$.

We root $\cT$ at a node $r \in V(T)$ so that the bags of each component of $T - r$ contain at most $kn/2$ vertices of $G \setminus \bag(r)$.
For a node $t \in V(T)$, we let $T_t$ be the subtree rooted at $t$, and $V_t$ the union of bags of $T_t$.
Because at most $k^2$ vertices of $G-\bag(r)$ are contained in small components and $kn - kn/2 - (k+1) > k^2$, there are at least two children $c$ of $r$ so that $V_c$ contains a large component of $G-\bag(r)$.
For both of them it holds that $|V_c| \ge kn - kn/2 - (k+1) - k^2 \ge kn/4$.
By $kn/4 > k^2 + k + 1$, this again implies that $V_c \setminus \bag(c)$ contains a large component of $G-\bag(c)$.

We say that a non-root node $t \in V(T)$ is \emph{potent} if $V_t \setminus \bag(t)$ contains a large component of $G-\bag(t)$.
We argued above that exactly two children of the root are potent, denote them by $c_1$ and $c_2$.
Let $t$ be a descendant of $c_1$ so that all ancestors of $t$ are potent (including $t$ itself) but none of the children of $t$ are potent.

\begin{claim}
\label{lem:gridlb:claim4}
The path from $t$ to $c_1$ contains $\Omega(kn)$ nodes with bags of size $\ge k$.
\end{claim}
\begin{claimproof}
For each node $x$ on this path, the graph $G-\bag(x)$ has two large components: One contained in $V_x \setminus \bag(x)$, and one contained in $V_{c_2} \setminus \bag(c_2)$.
\Cref{lem:gridlb:claim2} implies that all bags on the path have size $\ge k$.
It remains to prove that this path contains $\Omega(kn)$ nodes.

Let $x$ and $y$ be on the path so that $y$ is a child of $x$.
We have that $|\bag(x) \setminus \bag(y)| \le 1$, because either $\bag(x) \subseteq \bag(y)$, or $\bag(y) \subseteq \bag(x)$, $|\bag(y)| \ge k$, and $|\bag(x)| \le k+1$.
Furthermore, because $G-\bag(x)$ has two large components, one contained in $V_{c_2} \setminus \bag(c_2)$ and one in $V_y \setminus \bag(y)$, the subtrees of the other children of $x$ contain in total at most one vertex not in $\bag(x)$.
It follows that $|V_y| \ge |V_x|-2$.

Let $s$ be the child of $t$ so that $V_s \setminus \bag(t)$ contains a large component of $G - \bag(t)$.
Because $V_s \setminus \bag(s)$ does not contain a large component of $G - \bag(s)$, we have that $|V_s| \le k^2 + k+1$.
It follows that $|V_t| \le 2 k^2 + 2k + 2$.
We have that $|V_{c_1}|-|V_t| \ge \Omega(kn)$, but $|V_y| \ge |V_x|-2$ for all consecutive $x$ and $y$ on the path, implying that the path contains $\Omega(kn)$ nodes.
\end{claimproof}

\Cref{lem:gridlb:claim4} finishes the proof.
\end{proof}

\section{Conclusions}
In this paper we gave a method for designing TLFPT algorithms parameterized by treewidth, giving in particular a TLFPT version of Courcelle's theorem.
This solved three questions posed by Bumpus et al.~\cite{DBLP:journals/corr/abs-2606-02492}, and partially resolved a fourth question.
We believe that our method applies for most decision problems and unweighted optimization problems that are solved by dynamic programming on tree decompositions.

A setting for which it is not clear whether TLFPT algorithms can be obtained is weighted optimization problems, even when the weights are in $[n]$.
In particular, we ask as an open question whether maximum weight independent set parameterized by treewidth, with weights in $[n]$, where $n$ is the number of vertices, is in TLFPT.
The main barrier is solving the problem in TLFPT time when parameterized by the maximum size of a connected component; we believe that if a TLFPT algorithm existed with this parameterization, it could also be lifted to treewidth with our techniques.

We gave a TLFPT algorithm for computing the value of treewidth, but noted that an explicit representation of an optimum-width tree decomposition may require space $\Omega(k(n+m))$.
However, it could still be possible that an implicit representation could be computed in TLFPT, or that an explicit representation for an approximation better than $2^{\OO(k)}$ could be computed in TLFPT time.
Perhaps the techniques of Boja{ń}czyk and Pilipczuk~\cite{DBLP:journals/lmcs/BojanczykP22} could be useful to this end.

The two open questions from~\cite{DBLP:journals/corr/abs-2606-02492} that we did not address are (1) ``Which FPT problems are likely not to be in TLFPT?'' and (2) ``Is rankwidth TLFPT parameterized by rankwidth?''.
For the second question, we note that if the input graph is sufficiently dense, then the algorithm of Korhonen and Sokolowski~\cite{DBLP:conf/stoc/Korhonen024} runs in TLFPT time, but for sparse input graphs, it is open even whether rankwidth is linear FPT parameterized by rankwidth.
We agree that designing techniques for ruling out TLFPT algorithms, especially for problems known to be linear FPT, is an interesting open direction.



\bibliographystyle{alpha}
\bibliography{main}

\end{document}